%% file: main.tex
\documentclass{article}

\usepackage{authblk}

\usepackage{latexsym}
\usepackage{amssymb}
\usepackage{amsmath}
\usepackage{amsthm}
\usepackage{booktabs}
\usepackage{enumerate}
\usepackage{graphicx}
\usepackage{color}
\usepackage{algorithm}
\usepackage[]{algpseudocode}
\usepackage{varwidth}
\usepackage{tikz}
\usetikzlibrary{automata, arrows.meta, bbox, calc, positioning, shapes.geometric, patterns}
\usepackage{cleveref}
\usepackage{circuitikz}
\usepackage{paralist}

\input{mathstyle}

\input{macro}

\input{fig}

\newtheorem{theorem}{Theorem}
\newtheorem{lemma}[theorem]{Lemma}

\newtheorem{definition}{Definition}
\newtheorem{example}{Example}
\newtheorem{construction}{Construction}

\begin{document}







\title{Synthesis of Reward Machines\\for Multi-Agent Equilibrium Design\\ (Full Version)~\thanks{A conference version of this work will appear in the proceedings of the 27th European Conference on Artificial Intelligence (ECAI'24).}}

\author[1]{Muhammad Najib}
\author[2]{Giuseppe Perelli}

\affil[1]{Heriot-Watt University}
\affil[2]{Sapienza University of Rome}

\date{}

\maketitle



\begin{abstract}
Mechanism design is a well-established game-theoretic paradigm for designing games to achieve desired outcomes. This paper addresses a closely related but distinct concept, equilibrium design. 
Unlike mechanism design, the designer's authority in equilibrium design is more constrained; she can only modify the incentive structures in a given game to achieve certain outcomes without the ability to create the game from scratch.
We study the problem of equilibrium design using dynamic incentive structures, known as reward machines.
We use weighted concurrent game structures for the game model, with goals (for the players and the designer) defined as mean-payoff objectives. We show how reward machines can be used to represent dynamic incentives that allocate rewards in a manner that optimises the designer's goal.
We also introduce the main decision problem within our framework, the payoff improvement problem. 
This problem essentially asks whether there exists a dynamic incentive (represented by some reward machine) that can improve the designer's payoff by more than a given threshold value.
We present two variants of the problem: strong and weak. We demonstrate that both can be solved in polynomial time using a Turing machine equipped with an NP oracle. Furthermore, we also establish that these variants are either NP-hard or coNP-hard.
Finally, we show how to synthesise the corresponding reward machine if it exists.
\end{abstract}




\input{intro}

\input{prelims}

\input{rw_machines}

\input{rw_eng}

\input{algo}

\input{conclusion}



\section*{Acknowledgments}
Perelli was supported by the PNRR MUR project PE0000013-FAIR and the PRIN 2020 projects PINPOINT. He was also supported by Sapienza University of Rome under the ``\emph{Progetti Grandi di Ateneo}'' programme, grant RG123188B3F7414A (ASGARD - Autonomous and Self-Governing Agent-Based Rule Design).




\bibliographystyle{plain}
\bibliography{biblio}

\clearpage
\appendix
\input{supplement}

\end{document}

%% file: mathstyle.tex
\usepackage{xargs}

\usepackage{xspace}

\usepackage{xstring}

\usepackage{boolexpr}

\usepackage[cmex10]{mathtools}
\usepackage{latexsym}
\usepackage{textcomp}
\usepackage{pifont}

\newcommand{\argemp}[2]{\if&#1&\else#2\fi}

\newcommand{\argdef}[2]{\if&#1&#2\else#1\fi}

\newcommand{\argint}[3]{\if&#2&\else#1#2#3\fi}

\newcommand{\argext}[3]{\if&#1&#3\else#1\if&#3&\else#2#3\fi\fi}

\newcommandx{\mthfnt}[3][1=, 2=0]{{
	\IfStrEqCase{#1}
	{%
		{}%
		{#3}%
		{Name}%
		{%
			\IfStrEqCase{#2}
			{%
				{0}{\mathcal{#3}}%
				{1}{\mathscr{#3}}%
				{2}{\mathfrak{#3}}%
				{3}{\mathbb{#3}}%
			}
			[\ensuremath{\clubsuit}]%
		}%
		{Set}%
		{%
			\IfStrEqCase{#2}
			{%
				{0}{\mathrm{#3}}%
				{1}{\mathsf{#3}}%
				{2}{\mathbb{#3}}%
				{3}{\mathbf{#3}}%
			}
			[\ensuremath{\clubsuit}]%
		}%
		{Fun}%
		{%
			\IfStrEqCase{#2}
			{%
				{0}{\mathsf{#3}}%
				{1}{\mathrm{#3}}%
			}
			[\ensuremath{\clubsuit}]%
		}%
		{Rel}%
		{%
			\IfStrEqCase{#2}
			{%
				{0}{\mathit{#3}}%
				{1}{\mathtt{#3}}%
			}
			[\ensuremath{\clubsuit}]%
		}%
		{Sym}%
		{%
			\IfStrEqCase{#2}
			{%
				{0}{\mathtt{#3}}%
				{1}{\mathbf{#3}}%
			}
			[\ensuremath{\clubsuit}]%
		}%
		{Elm}%
		{\mathnormal{#3}}
	}
[\ensuremath{\clubsuit}]%
}}

\newcommand{\mthsub}[1]{\argemp{#1}{\ensuremath{_{\mathnormal{#1}}}}}

\newcommand{\mthsup}[1]{\argemp{#1}{\ensuremath{^{\mathnormal{#1}}}}}

\newcommandx{\mth}[5][1=, 2=0, 4=, 5=]{{\ensuremath{\mthfnt[#1][#2]{#3}\mthsub{#4}\mthsup{#5}}}}

\newcommandx{\mtharg}[6][1=, 2=0, 4=, 5=]{{\mth[#1][#2]{#3}[#4][#5]\ensuremath{\argint{(}{#6}{)}}}}

\newcommand{\mthempty}{\mth[][]}

\newcommand{\mthstyname}{0}
\newcommand{\mthname}[1][]{\mth[Name][\argdef{#1}{\mthstyname}]}

\newcommand{\mthstyset}{0}
\newcommand{\mthset}[1][]{\mth[Set][\argdef{#1}{\mthstyset}]}

\newcommand{\mthstyfun}{0}
\newcommand{\mthfun}[1][]{\mth[Fun][\argdef{#1}{\mthstyfun}]}

\newcommand{\mthstysym}{0}
\newcommand{\mthsym}[1][]{\mth[Sym][\argdef{#1}{\mthstysym}]}

\newcommand{\mthstyelm}{0}
\newcommand{\mthelm}[1][]{\mth[Elm][\argdef{#1}{\mthstyelm}]}

\newcommand{\tuple}[1]
{\ensuremath{\!\argint{\langle}{#1}{\rangle}}}

%% file: macro.tex
\newcommand{\ignore}[1]{}

\newcommand{\SetN}{\mathbb{N}}
\newcommand{\SetZ}{\mathbb{Z}}

\newcommand{\card}[1]{\mthempty{\argint{\vert}{#1}{\vert}}}
\newcommand{\size}[1]{\mthempty{\argint{\vert\vert}{#1}{\vert\vert}}}

\newcommand{\LTL}{\mthfun{LTL}\xspace}

\def\Nat{\mathbb{N}}

\newcommand{\Ag}{\mthset{N}}
\newcommand{\Ac}{\mthset{Ac}}
\newcommand{\AcProf}{\vec{\Ac}}
\newcommand{\St}{\mthset{St}}
\newcommand{\state}{\mthsym{s}\xspace}
\newcommand{\instate}{\ensuremath{\state_{\mthfun{in}}}}

\renewcommand{\Game}{\mthname{G}}

\newcommand{\trnFun}{\mthfun{tr}}
\newcommand{\act}{\mthsym{a}}
\newcommand{\jact}{\vec{\mthsym{a}}}

\newcommand{\StrSet}{\mthset{Str}}

\newcommand{\strElm}{\sigma}

\newcommand{\strpElm}{\mthelm{\vec{\sigma}}}

\newcommand{\NE}{\mthset{NE}}

\newcommand{\winsym}{\mthset{Win}}
\newcommandx{\Win}[3][1=, 2=, 3=]
{\mthset{\winsym#3}[#1][#2]}

\newcommand{\presym}{\mthfun{Pre}}
\newcommandx{\Pre}[3][1=, 2=, 3=]
{\mthset{\presym#3}[#1][#2]}

\newcommand{\eqsym}{\mthfun{Eq}}
\newcommandx{\Eq}[3][1=, 2=, 3=]
{\mthset{\eqsym#3}[#1][#2]}

\newcommandx{\AFW}[5][1=, 2=, 3=, 4=, 5=]
{\txtargname{AFW#5{\small\argint{$[$}{#1}{$]$}}}[#2][#3]{#4}\xspace}

\newcommand{\MP}[1][]{%
	\ifthenelse{\equal{#1}{}}{{\small{\sf mp}}}{{\small{\sf mp}}(#1)}%
\xspace}

\def\Rat{\mathbb{Q}}

\providecommand{\strFun}[1][]{\mthfun{\sigma}}
\providecommand{\pstrFun}[1][]{\mthfun{\sSym}}

\newcommand{\wFun}{\mthfun{w}}

\def\avg{{\sf avg}}

\def\trg{{\sf trg}}

\newcommand{\pay}{\mthfun{pay}\xspace}
\newcommand{\bestNE}{\mthfun{bestNE}\xspace}
\newcommand{\worstNE}{\mthfun{worstNE}\xspace}

\def\ptime{\mthfun{P}\xspace}
\def\np{\mthfun{NP}\xspace}
\def\conp{\mthfun{coNP}\xspace}

\def\FP{\mthfun{FP}\xspace}

\newcommand{\sink}{\mthsym{sink}}

\newcommand{\DeltaPTwo}{\Delta^{\mthfun{P}}_2}

\newcommand{\RM}{\mathcal{M}\xspace}

\newcommand{\protocol}{\mthfun{d}}

\newcommand{\implement}{\dagger}

\newcommand{\out}{\mthfun{out}}
\newcommand{\Paths}{\mthfun{Paths}}

\newcommand{\MinW}{\mthset{MinW}}
\newcommand{\MaxW}{\mthset{MaxW}}

\renewcommand{\epsilon}{\varepsilon}
\newcommand{\epsilondash}{\epsilon\text{-}}
\newcommand{\GammaNE}{\Gamma\mthfun{NE}\xspace}

%% file: fig.tex
\usepackage{tikz}
\usetikzlibrary{arrows,shapes,snakes}

\usepackage{wrapfig}

\newcommand{\NewExample}{
\scalebox{0.6}{
	\vspace{-15pt}
	\begin{minipage}{0.25\textwidth}
		\begin{circuitikz}
			\draw[pattern = dots, pattern color=red, draw=white] (0,0) rectangle (3,1);
			\draw (0,0) rectangle (3,5);
			\fill[black, draw=black] (1,1) rectangle (2,2);
			\fill[black, draw=black] (1,3) rectangle (2,4);

			\draw[brown, thick] (2,1) to[normal open switch] (3,1);
			\draw[brown, thick] (3,4) to[normal open switch] (2,4);
			\draw[brown, thick] (2,2) to[normal open switch] (2,3);
			
			\draw[-stealth, blue] (2.5,1) -- (2.5,1.5);
			\draw[-stealth, blue] (2.5,4) -- (2.5,3.5);
			\draw[-stealth, blue] (2,2.5) -- (1.5,2.5);
			
			\node at (.5,2.5) {$ s $};
			\node at (1.5,4.5) {$ l $};
			\node at (1.5,.5) {$ r $};
			\node at (2.5,2.5) {$ m $};
			
			\node at (0.5,1.5) {\includegraphics[width=0.8cm]{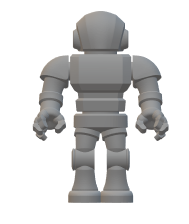}};
		\end{circuitikz}
	\end{minipage}
	\hspace{20pt}
	\begin{minipage}{0.25\textwidth}
	\scalebox{0.8}{
	\begin{tikzpicture}[state/.style={circle, draw, minimum size=1cm}, node distance=2cm]
		\node[state, label=below:{$ $}] (s^0) {\small $t$};
		\node[] (start) [left =0.5cm of s^0] {};
		\node[state, label=below:{$ $}] (s^2) [right = of s^0] {\small $m$};
		\node[state, label=below:{$ $}] (s^1) [above = of s^2] {\small $l$};
		\node[state, label=below:{$ $}] (s^3) [below = of s^2] {\small $r$};
		
		\draw [-{Latex[width=2mm]}]
		(start) edge node{} (s^0)
		(s^0) edge[bend right] node[right]{\small $L$} (s^1)
		(s^1) edge[bend right] node[right]{\small $T$} (s^0)
		
		(s^2) edge[] node[above]{\small $T$} (s^0)
		
		(s^0) edge[bend left] node[right]{\small $R$} (s^3)
		(s^3) edge[bend left] node[right]{\small $T$} (s^0)
		
		(s^1) edge[] node[right]{\small $M$} (s^2)
		(s^3) edge[] node[right]{\small $M$} (s^2)
		
		;
	\end{tikzpicture}
	}
	\end{minipage}
}
}

\newcommand{\RMExample}{
\scalebox{0.65}{

\begin{tikzpicture}[state/.style={diamond, draw, minimum size=1cm}, node distance=1cm, rotate=-90]
	\node[state, label=below:{$ $}] (s^0) {\small $q_0$};
	\node[] (start) [above =0.5cm of s^0] {};
	\node[state, label=below:{$ $}] (s^3) [below = of s^0] {\small $q_3$};
	\node[state, label=below:{$ $}] (s^1) [left = of s^3] {\small $q_1$};
	\node[state, label=below:{$ $}] (s^2) [right = of s^3] {\small $q_2$};

	\draw [-{Latex[width=2mm]}]
	(start) edge node{} (s^0)
	(s^0) edge[] node[left]{\small $l, 0$} (s^1)
	(s^0) edge[] node[right]{\small $r, 0$} (s^2)
	(s^1) edge[] node[below]{\small $m, 1$} (s^3)
	(s^2) edge[] node[below]{\small $m, 0$} (s^3)
	(s^3) edge[] node[left]{\small $t, 0$} (s^0)
	
	(s^1) edge[bend left] node[left]{\small $t, 0$} (s^0)
	(s^2) edge[bend right] node[right]{\small $t, 0$} (s^0)
	
	;
\end{tikzpicture}

}
}

\newcommand{\RMExampletwo}{
	\scalebox{0.7}{
		
		\begin{tikzpicture}[state/.style={diamond, draw, minimum size=1cm}, node distance=0.9cm]
			\node[state, label=below:{$ $}] (s^0) {\small $q_0$};
			\node[] (start) [above =0.5cm of s^0] {};
			\node[state, label=below:{$ $}] (s^1) [left = of s^0] {\small $q_1$};
			\node[state, label=below:{$ $}] (s^3) [left = of s^1] {\small $q_3$};
			\node[state, label=below:{$ $}] (s^2) [right = of s^0] {\small $q_2$};
			\node[state, label=below:{$ $}] (s^4) [left = of s^3] {\small $q_4$};
			\node[state, label=below:{$ $}] (s^5) [left = of s^4] {\small $q_5$};
			\node[state, label=below:{$ $}] (s^6) [left = of s^5] {\small $q_6$};

			\draw [-{Latex[width=2mm]}]
			(start) edge node{} (s^0)
			(s^0) edge[] node[above]{\small $l, 0$} (s^1)
			(s^0) edge[] node[above]{\small $r, 0$} (s^2)
			(s^1) edge[] node[above]{\small $m, 0$} (s^3)
			(s^3) edge[] node[above]{\small $s, 0$} (s^4)
			(s^4) edge[] node[above]{\small $l, 0$} (s^5)
			(s^5) edge[] node[above]{\small $m, 1$} (s^6)
			
			(s^1) edge[bend right] node[below]{\small $s, 0$} (s^0)
			(s^5) edge[bend left] node[below]{\small $s, 0$} (s^0)
			(s^6) edge[bend left] node[above]{\small $s, 0$} (s^0)
			
			(s^2) edge[loop, out=70, in=110, distance=1cm] node[above]{$*, 0$} (s^2)
			
			;
		\end{tikzpicture}

	}
	\vspace{-15pt}
}

%% file: intro.tex
\section{Introduction}
	Over the past decade, Nash equilibrium (NE) and other game-theoretic concepts have been extensively used to analyse concurrent and multi-agent systems (see e.g.,~\cite{GHW15,bouyer2015pure,WooldridgeGHMPT16,GutierrezHW17-aij}).
	In this research, systems are modelled as games with agents acting rationally to fulfil their preferences. While preferences are often expressed qualitatively (e.g. by temporal logic formulae), many systems require more complex models to capture quantitative aspects like resource consumption, cost, or performance~\cite{GMPRW17,AKP18,AL20,BG22}. 
	Games with \textit{mean-payoff} objectives~\cite{ZP96} provide such richer preference modelling.

	The game-theoretical analysis of mean-payoff games (MPGs) is a significant research area, especially in verifying their correctness~\cite{UW11,brenguier2016robust,BriceRB21,steeples2021mean,gutierrez2023complexity,brice2023rational}.
	This involves checking whether a formal property is satisfied in some or all equilibrium outcomes.
	A pertinent question is: ``what if the property is not satisfied in any equilibrium outcomes?''
	\textit{Equilibrium design}~\cite{GNPW19} addresses this question.
	Inspired by the mechanism design paradigm~\cite{Mye89,HR06}, equilibrium design offers a way to rectify equilibrium outcomes.
	However, unlike mechanism design, the designer in equilibrium design cannot create the game from scratch, but can only modify the incentive structures of an existing game.

	In~\cite{GNPW19}, the authors proposed \textit{subsidy schemes} to introduce equilibria in concurrent MPGs satisfying some LTL formula~\cite{pnueli:77a}.
	In that setting, a subsidy scheme is modelled by a function mapping from states and players to additional rewards.
	If a player visits a certain state, the corresponding reward is paid to the player.
	In this paper, we generalise this incentive model with \textit{reward machines}~\cite{icarte2022reward}.
	Such machines implement a reward mechanism that considers the execution history to dynamically assign rewards.
	Thus, at each iteration of the game, every agent receives a utility combining the original weight and an authoritative reward based on the current game state and the internal reward machine state.
	As we will show later (see \Cref{ex:1}), this generalisation allows us to obtain a more expressive model of incentive.

	We consider games where each agent has a weight function over the states, with mean-payoff aggregation as their utility function.
	Additionally, a global weight function measures the designer satisfaction, also as a mean-payoff value over executions.
	We employ reward machines to improve the designer satisfaction.
	Intuitively, these machines reconfigure weights after each iteration, thus reshaping the set of equilibria.
	The objective is to improve the global payoff over the set of equilibria by a fixed amount $\Delta$.
	This can be achieved \emph{strategically} by synthesising and implementing an appropriate reward machine.
	To make the setting realistic, we assume the reward spent on each agent in every iteration is subtracted from the global weight, factoring the cost into the resulting global payoff.
	
	
	In MPGs, infinite memory may be required to achieve optimal values~\cite{VCDHRR15}. As we will demonstrate later, infinite memory may be necessary to achieve the optimal global payoff value. However, since a reward machine is typically represented by a finite-state machine (specifically, a Mealy machine in this work), there may be cases where no finite-state reward machine can improve the global payoff by a given $ \Delta $. Therefore, we consider an approximate solution. In particular, we aim to find a reward machine that can improve the global payoff by a value $ \epsilon $-close to $ \Delta $ for a given $ \epsilon > 0 $. Moreover, $ \epsilon $ can be arbitrarily small, allowing for an arbitrary level of precision.
	
	In general, a game may have multiple equilibria. Therefore, we study the problem under both the \textit{optimistic} and \textit{pessimistic} settings. Specifically, we consider the problem of improving the global mean-payoff over the \textit{best} possible NE by adopting an optimistic view that agents will select the equilibrium \textit{most} convenient for the designer. We call this the \textit{weak improvement} problem. 
	Conversely, we also consider improving the global mean-payoff over the \textit{worst} possible NE, considering the pessimistic case when agents select the \textit{least} convenient equilibrium for the designer. We call this the \textit{strong improvement} problem.
	Furthermore, we classify the complexity of these problems. We show that both can be solved in $\ptime^{\np} = \DeltaPTwo$ and are at least $\np$-hard or $ \conp $-hard. To our knowledge, this is the first work that employs reward machines in the context of MPGs and game-theoretic equilibria.
	
	\paragraph{Related work} 
	As previously mentioned, this work is closely related to \cite{GNPW19}, but it differs in several key aspects. Firstly, our incentive model is more expressive due to the use of reward machines. Furthermore, we measure the global property using a quantitative metric (i.e. mean-payoff value), as opposed to the qualitative property in \cite{GNPW19}. In this respect, we provide a richer modelling of global preferences. Equilibrium design has a deep connection to mechanism design, but the two are not exactly the same. Typically in mechanism design, the designer is not given a predetermined game structure, but instead is required to provide one. Moreover, in mechanism design, the designer must ensure the reward structure is \textit{incentive compatible}, which is not the case in equilibrium design, as the designer is primarily interested in the global payoff, unlike in mechanism design where the designer is primarily interested in the agents' payoffs.
	
	On the other hand, the concept of a reward machine originated from the field of reinforcement learning (RL). Much of the existing work is within the domain of single-agent RL \cite{icarte2018using,toro2019learning,icarte2022reward}. In \cite{neary2021}, the authors explored reward machines for multi-agent RL systems. However, in this work, the reward machine is manually generated, as opposed to being automatically synthesised. \cite{varricchione2023synthesising} tackles the problem of automatically synthesising reward machines in cooperative multi-agent RL. Specifically, the reward machines are partly synthesised from Alternating-time Temporal Logic (ATL) specifications. However, this line of research focuses on RL systems, which differ from MPGs. Moreover, none of these papers consider any game-theoretical solution concepts.
	
	Another related line of work involves designing equilibria using \textit{norms}. Norm-based mechanism design has been studied in \cite{bulling2016norm}. In particular, they studied \textit{weak} and \textit{strong} implementability, which are related to the problems addressed in our work in the sense that they correspond to optimistic (\textit{``there is some good behaviour''}) and pessimistic (\textit{``all behaviours must be good''}) assumptions. In \cite{huang16,perelli19,ADLP22}, automata-based norms, referred to as \textit{dynamic norms}, are considered. All of these works fall within the domain of normative systems, which is different from the setting considered in this paper. We believe that an incentive-based equilibrium design provides a complementary approach to norm-based equilibrium design. This is because in some circumstances, a norm may not be enforceable, but only incentives are possible (e.g., congestion/road pricing in the Ultra Low Emission Zone (ULEZ) in London).
	

%% file: prelims.tex
\section{Preliminaries}


In this section we introduce the basic notions that will be used throughout the paper. We start with the definition of mean-payoff value and multi-player mean-payoff games.

\paragraph{Mean-Payoff}
For an infinite sequence $r \in \mathbb{R}^\omega$, let $\MP(r)$ be
the \emph{mean-payoff} value of $r$, that is, 
$ \MP(r) = \lim \inf_{n \to \infty} \avg_n(r) $
where, for $n \in \mathbb{N}\setminus\{0\}$, we define
$\avg_n(r) = \frac{1}{n}\sum_{j=0}^{n-1} r_j$, with $r_j$ the $(j\!+\!1)$th element of $r$.

\paragraph{Multi-Player Mean-Payoff Game}
A \emph{multi-player mean-payoff game} is a tuple
$ \Game = \tuple{\Ag,  \Ac, \St, \instate, (\protocol_i)_{i \in \Ag}, \trnFun, (\wFun_{i})_{i \in \Ag}, \wFun_{g} } $ 
where 
\begin{itemize}
	\item $\Ag = \{1,\dots,n\}$, $\Ac$, and $\St$ are finite non-empty sets of \emph{players}, \emph{actions}, and \emph{states}, respectively;
	\item $\instate \in \St$ is the \emph{initial state};
	\item $ \protocol_i : \St \to 2^{\Ac} \setminus \{\varnothing\} $ is a \textit{protocol function} for player $ i $ returning possible action at a given state;
	\item $\trnFun : \St \times \AcProf \rightarrow \St$ is a \emph{transition function} mapping each pair consisting of a state $s \in \St$ and an \emph{action profile} $\jact = (\act_{1}, \dots, \act_{n}) \in \AcProf = \Ac^{n}$, one for each player, to a successor state---we write $\jact_{i}$ for $\jact_{\{i\}}$ and $ \jact_{-i} $ for $ \jact_{\Ag \setminus \{i\}} $.
	For two decisions $\jact$ and $\jact'$, we write $(\jact_{C}, \jact_{-C}')$ to denote the decision where the actions for players in $ C \subseteq \Ag $ are taken from $\jact$ and the actions for players in $ \Ag \setminus C $ are taken from $\jact'$;
	\item  $\wFun_{i}: \St \to \SetZ$ is player $ i $'s \textit{weight function} mapping, for every player~$i$ and every state of the game into an integer number; and
	\item $ \wFun_{g}: \St \to \SetZ $ is a \textit{global weight function} mapping every state of the game into an integer number.
\end{itemize}
We define the \textit{minimum} and \textit{maximum weights} appearing in $ \Game $ as follows.

\begin{definition}
	\label{def:minmaxw}
	For a given game $ \Game $ and its set of states $ \St $, define $ \MinW^{\Game}_{j} = \min \{ \wFun_{j}(\state) \mid \state \in \St \} $ and $ \MaxW^{\Game}_{j} = \max \{ \wFun_{j}(\state) \mid \state \in \St \} $.
\end{definition}


A \textit{path} is an infinite sequence $ \pi = \state_0, \state_1, \state_2,... \in \St^{\omega} $ such that for each $ k \in \Nat $, there is an action profile (in $ k $-th step) $ \jact^{k} \in \prod_{i \in \Ag} \protocol_i(\state_k) $, such that
$ \state_{k + 1} = \trnFun(\state_k, \jact^{k}) $. We write $ \pi_{\leq k} $ to denote the prefix of $ \pi $ up to and including $ \state_{k}$. Similarly, $ \pi_{\geq k} $ denotes the suffix of $ \pi $ starting from $\state_{k} $. Let $ \Paths_{\Game}(s) $ be the set of all possible paths in $ \Game $ starting from $ s $.

A \textit{strategy} for agent $ i $ is a Mealy machine $\strElm_{i} = (T_{i}, t_{i}^{0}, \St, \gamma_i, \rho_i) $, where $T_{i}$ is
a finite and non-empty set of \emph{internal states}, $ t_{i}^{0} $ is
the \emph{initial state},
$\gamma_i: \St \times T_{i} \rightarrow T_{i} $ is a deterministic
\emph{internal transition function}, and
$\rho_i: \St \times T_{i} \rightarrow \Ac_i$ an \emph{action function}.
We say that a strategy $ \strElm_{i} $ is \textit{valid} with respect to $ \Game $ if and only if $ \rho_i(s,t_j) \in \protocol_i(s) $. From now on, we restrict our attention to valid strategies, and, unless otherwise stated, refer to them simply as strategies.
We denote $\StrSet_i(\Game)$ the set of valid strategies for player $i$ in $\Game$.
Moreover, for a given strategy $ \strElm_{i} $ and a finite sequence $ \hat{\pi} \in \St^{\ast} $, by $ \strElm_{i}(\hat{\pi}) \in \Ac $ we denote the action prescribed by the action function $ \rho_i $ of $ \strElm_{i} $ after the sequence $ \hat{\pi} $ has been fed to the internal transition function $ \gamma_i $.
Note that the model of strategies implies that strategies have \textit{perfect information} and \textit{finite memory}, although we impose no bounds on memory size.

A \emph{strategy profile} $\strpElm = (\strElm_1, \dots, \strElm_n)$ is a vector of strategies, one for each player.
As with actions, $\strpElm_{i}$ denotes the strategy assigned to player $i$ in profile $\strpElm$.
Moreover, by $(\strpElm_{B}, \strpElm'_{C})$ we denote the combination
of profiles where players in disjoint $B$ and $C$ are assigned their corresponding strategies in $\strpElm$ and $\strpElm'$, respectively.
We denote $\StrSet_{A}(\Game)$ the set of strategy profiles for the set $A$ of agents.
We also use $\StrSet(\Game) = \StrSet_{\Ag}(\Game)$ to denote the strategy profiles for all the agents in the game.
Whenever the game is clear from the context, we also simply use $\StrSet$.
Once a state $s$ and profile $\strpElm$ are fixed, the game has an \emph{outcome}, a path in $ \Game $, denoted by $\pi(\strpElm, s)$. 
Because strategies are deterministic, $\pi(\strpElm, s)$ is the unique path induced by $\strpElm$, that is, the sequence $\state_0, \state_1, \state_2, \ldots$ such that 
\begin{itemize}
	\item $\state_{k + 1} = \trnFun (\state_k, (\rho_1(\state_k,t^k_1), \ldots, \rho_n(\state_k,t^k_n)))$, and 
	\item $t^{k + 1}_i = \gamma_i(s^k_i, t^k_i)$, for all $k \geq 0$. 
\end{itemize}

For a subset of agents $ C \subseteq \Ag $ and strategies $ \strpElm_{C} $, we say that a path $ \pi $ is \textit{compatible} with $ \strpElm_{C} $ if, for every $ k \in \Nat $, there exists an action profile $ \jact^{k} $ with $ \act_i^{k} = \strElm_{i}(\pi_{\leq k}) $ for each $ i \in C $, such that $ \state_{k + 1} = \trnFun(\state_k,\jact^{k}) $. Intuitively, $ \pi $ is compatible with $ \strpElm_{C} $ if it can be generated when the agents in $ C $ play according to their respective strategies. We denote by $ \out_{\Game}(\state,\strpElm_{C}) $ the set of paths starting from $ s $ and compatible with $ \strpElm_{C} $. Observe that $ \Paths_{\Game}(s) $ can also be written as $ \out_{\Game}(\state,\varnothing) $. 

Given a game $ \Game $ and a strategy profile $ \strpElm $, a path $\pi(\strpElm)$ induces, for each player $ i $, an infinite sequence of integers $\wFun_i(\pi(\strpElm)) = \wFun_i(\instate) \wFun_i(\state_1) \cdots$. Similarly, $\pi(\strpElm)$ also induces such a sequence of integers for the global weight function $ \wFun_{g}(\cdot) $.
The \textit{payoff} of player $i$ in game $\Game$ is $\pay_{i}^{\Game}(\strpElm) = \MP(\wFun_{i}(\pi(\strpElm)))$, and the \textit{global payoff} of $\Game$ is $\pay_{g}^{\Game}(\strpElm) = \MP(\wFun_{g}(\pi(\strpElm)))$.
Whenever the game is clear from the context, we simply use $\pay_{i}(\strpElm)$ and $\pay_{g}(\strpElm)$, respectively. 

\paragraph{Nash Equilibrium}
Using payoff functions, we can define the game-theoretic concept of Nash equilibrium~\cite{OR94}. 
For a multi-player game $\Game$, a strategy profile
$\strpElm$ is a \emph{Nash equilibrium} of~$\Game$ if, for every player~$i$ and strategy $\strElm'_i$ for player~$i$, we have
$$
\pay_i(\strpElm) \geq	\pay_i((\strpElm_{-i},\strElm'_i)) \ . 
$$
We also say that $\strpElm$ is a $j$-fixed Nash Equilibrium~\cite{KPV14} if
$\pay_i(\strpElm) \geq	\pay_i((\strpElm_{-i},\strElm'_i))$ for every player $i \neq j$ different from the fixed $j$.

Let $\NE(\Game)$  and $\NE_{j}(\Game)$ be the set of Nash Equilibria and $j$-fixed Nash Equilibria of~$\Game$.
We define $\bestNE(\Game) = \sup_{\strpElm \in \NE(\Game)} \{\pay_{g}(\strpElm)\}$ as the \emph{best global payoff} over the set of possible outcomes sustained by a Nash Equilibrium in the game.
Equivalently, we define $ \worstNE(\Game) = \inf_{\strpElm \in \NE(\Game)}\{\pay_{g}(\strpElm)\} $ as the \emph{worst global payoff} over the set of possible outcomes sustained by a Nash Equilibrium in the game.
In the case of $ \NE(\Game) $ is empty, in order to make the values of $ \bestNE(\Game) $ and $ \worstNE(\Game) $ well defined, we assume that $ \bestNE(\Game) = \worstNE(\Game) = \MinW_{g}^{\Game} $.


%% file: rw_machines.tex
\section{Reward Machines for Equilibrium Design}


	In this section, we introduce a type of finite state machine, called a \emph{reward machine} (RM).
	A RM takes a path $ \pi $ as input, and outputs a sequence of vectors $\vec{v}_{0}, \vec{v}_{1} \cdots \in (\SetN^{n})^{\omega}$ that corresponds to the reward granted to the players at each step of the path. Formally, a RM is defined as a Mealy machine:
	
	\begin{definition}[Reward Machine]
		A RM is a Mealy machine $ \RM = \tuple{Q^{\RM},q_{0}^{\RM},\delta^{\RM},\tau^{\RM}} $, where $ Q^{\RM} $ is a finite (non-empty) set of states, $ q_{0}^{\RM} $ the initial state, $ \delta^{\RM} : Q^{\RM} \times \St \to Q^{\RM} $ a deterministic transition function, and $ \tau^{\RM} : Q^{\RM} \times \St \to \SetN^n $ a reward function where $\tau_{i}^{\RM}(q) = \tau^{\RM}(q)(i)$ is the reward in the form of a natural number $ k \in \SetN $ imposed on player $i$ if the play visits $(s, q) \in \St \times Q^{\RM}$.
		Sometimes, when it is clear from the context, the elements of the RM are denoted without superscripts.
	\end{definition}
	
	
	\paragraph*{Reward Machine implementation.} 
	
	For a given game $\Game = \tuple{\Ag,  \Ac, \St, \instate, (\protocol_i)_{i \in \Ag}, \trnFun, (\wFun_{i})_{i \in \Ag}, \wFun_{g} }$, the implementation of $\RM$ on $\Game$ is the game 
	
	\begin{center}	
	$\Game \implement \RM =\tuple{\Ag,  \Ac, \St \times Q, (\instate, q_{0}), (\protocol_i^{\RM})_{i \in \Ag}, \trnFun^{\RM}, (\wFun_{i}^{\RM})_{i \in \Ag}, \wFun_{g}^{\RM} }$, 
	\end{center}
	where:
	\begin{inparaenum}[(i)]
		\item 
			$\protocol_i^{\RM}(s, q) = \protocol_i(s)$, for each agent $i \in \Ag$;
		\item
			$\trnFun^{\RM}((s, q), \jact) = (\trnFun(s, \jact), \delta(s, q))$;
		\item 
			$\wFun_{i}^{\RM}(s, q) = \wFun_{i}(s) + \tau_{i}(s, q)$;
		\item 
			$\wFun_{g}^{\RM}(s, q) = \wFun_{g}(s) - \size{\tau(s, q)}$~\footnote{By $\size{\vec{v}} = \sum_{i \in \Ag}\card{v_i}$ we denote the classic Manhattan distance.}.
	\end{inparaenum}
	
	
	For a given natural number $\beta \in \SetN$, a $\beta$-RM, denoted $\RM_{\beta}$, is RM such that $\size{\tau(s,q)} \leq \beta$ for each $(s, q) \in \St \times Q$.
	In this paper, we consider a \emph{budget} $\beta$ being fixed and restrict our attention only to $\beta$-RMs.
	
	
	\begin{definition}[Global payoff improvement problems]
		For a given game $\Game$, a budget $\beta$, and a threshold $\Delta$.
		The \emph{global payoff weak improvement} problem consists in deciding whether there exists a $\beta$-RM $\RM$ such that:
		
		\begin{center}
			$\bestNE(\Game \implement \RM) - \bestNE(\Game) > \Delta$.
		\end{center}
		
	The \emph{global payoff strong improvement} problem consists in deciding whether there exists a $\beta$-RM $\RM$ such that:
	
	\begin{center}
		$\worstNE(\Game \implement \RM) - \worstNE(\Game) > \Delta$.
	\end{center}
	
	\end{definition}	

	Henceforth, for simplicity, we will use the term \emph{improvement problem} to refer to the global payoff improvement problem.
	
	
	At this point, it is important to note that the optimal values of $ \bestNE $ and $ \worstNE $ may not be achievable with finite-state strategies and reward machines. As such, to guarantee termination, we compute the approximate values instead. Moreover, our approach allows the values to be approximated to an arbitrary level of precision. We discuss this in detail in \Cref{sec:solving}.
	
	
	\paragraph*{Reward Machines vs Subsidy Schemes.}
		As previously mentioned, the reward model in this paper is a generalisation of the one considered in~\cite{GNPW19}, which is referred to as a \emph{subsidy scheme} in that paper.
		A subsidy scheme is defined as a function $ \kappa : \St \to \SetN^n $.
		This can be trivially expressed by a reward machine $ \RM = (Q^{\RM},q_{0}^{\RM},\delta^{\RM},\tau^{\RM}) $ where $ Q^{\RM} = \{q\}, q_{0}^{\RM} = q$, and for all $ \state \in \St, \delta^{\RM}(s,q) = q, \tau^{\RM}(s,q) = \kappa(s)$.
		In other words, subsidy schemes belong to the subclass of ``memoryless'' reward machines\footnote{We note that the semantics of the budget used here is slightly different to the one used in \cite{GNPW19}. In this work budget can be thought as ``capacity'' of additional reward in each time step, whereas in \cite{GNPW19} it is the total ``commitment'' of reward in the game.}. However, there are some cases in which memory is required. To illustrate this, consider the following simple example.
		
		\begin{figure}[t]
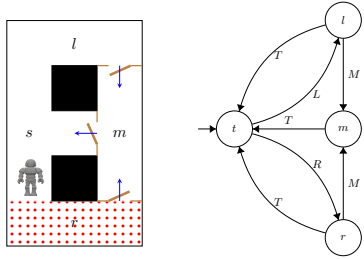

			\begin{center}
				\NewExample	
			\end{center}
			\caption{Graphical representation (left) and game arena (right) for \Cref{ex:1}.}
			\label{fig:ex1}
		\end{figure}

		\begin{example}\label{ex:1}
			Consider a scenario where a robot is situated in an environment shown in \Cref{fig:ex1} left, in which there are four locations $t, l, r, m$. 
			The robot can move from one location to another and is not allowed to stay in the same location for two consecutive time steps.
			There are three doors separating the locations, and they can only be passed according to their respective arrows. For instance, the robot can only move from $m$ to $t$ through the middle door, and not the other direction.
			Thus, the robot can reach $m$ from $t$ only through $ l $ or $ r $. 
			However, location $ r $ is still under maintenance, and it is best to avoid passing through it.
			Suppose that the designer wants to incentivise the robot to deliver goods from $t$ to $ m $ infinitely often.
			We can model this as a game $ \Game $ with $ \Ag = \{1\} $. 
			The game graph is shown in \Cref{fig:ex1} right. In each state, the actions available to the player correspond to the outgoing edges and their respective labels. 
			Let $t$ be the initial state, and $ \wFun_{1}(v) = 0 $ for all $ v \in \St = \{t, l, m, r \} $. Moreover, let $ \wFun_{g}(t) = \wFun_{g}(r) = 0, \wFun_{g}(l) = 1, \wFun_{g}(m) = 2 $; the designer receives reward of $ 2 $ when the robot visits $ m $ from $t$, furthermore, she receives extra reward of $ 1 $ when the robot uses corridor $ l $. Suppose that $ \beta = 1 $. Observe that $ \worstNE(\Game) = 0 $ corresponding to the sequence $ p(t, r)^{\omega} $ for some finite prefix $ p $.
			Suppose that we want to synthesise $ \RM $ such that given $ \Delta = \frac{1}{2} $, the strong improvement problem returns a positive answer. That is, $ \worstNE(\Game \implement \RM) - \worstNE(\Game) > \frac{1}{2} $. 
		\end{example}

		\begin{wrapfigure}[9]{r}{0.4\columnwidth}
			\centering
			\vspace{-1em}
			\RMExample
			\vspace{-1em}
			\caption{Reward machine $ \RM $.}
			\label{fig:rm}
		\end{wrapfigure}
		A reward machine that satisfies the constraint in \Cref{ex:1} is as follows: only give rewards of $ 1 $ when the robot visits $ m $ from $ l $. More formally, the reward machine is shown in \Cref{fig:rm} where $ \tau^\RM_{1}(q_1,m) = 1 $ and $ \tau^\RM_{1}(q, t) = 0 $ for all $ (q, t) \neq (q_1,m) $.
		The set of Nash equilibria in $ (\Game \implement \RM) $ corresponds to the sequence $ p(t, l, m)^{\omega} $ for some finite prefix $ p $. As such, we have $ \worstNE(\Game \implement \RM) = \frac{2}{3} $.
		Observe that such an incentive requires memory, since it needs to remember which path leading to $ m $ is taken by the robot. This is not possible with memoryless reward machine.

%% file: rw_eng.tex
\section{Reward Engineering}

\textcolor{red}
{
}

In this and the next section, we show how to solve global payoff improvement by constructing an \emph{auxiliary game} that allows to look at the problem as an equilibrium verification per se.
More specifically, such construction regards reward machines as the strategies of a designated agent in the game, whose weight function corresponds to the global weights of the original game updated with the rewards spent on the others at each iteration.
First we provide the definition of such auxiliary game, which is inspired from the constructions given in~\cite{perelli19, ADLP22}.

\begin{definition}
	Given a game $ \Game $ and a budget $ \beta \in \Nat $, we define its \textit{auxiliary} game $ \Game' = \tuple{\Ag', \Ac', \St', \instate', (\protocol_i)_{i \in \Ag'}, \trnFun', (\wFun_{i}')_{i \in \Ag'}} $, where 
	\begin{inparaenum}[(i)]
		\item $ \Ag' = \{0\} \cup \Ag, \Ac' = \Ac \cup \beta^n, \St' = \St \times \beta^n, \instate' = (\instate,\vec{0})$;
		\item $ \trnFun'((s,\vec{v}), (\jact,\vec{v}')) = (\trnFun(s,\jact), \vec{v}') $;
			\item $ \protocol_i'(s,\vec{v}) = \protocol_i(s), i \in \Ag $;
			\item $ \protocol_0'(s,\vec{v}) = \{ \vec{v} : \size{\vec{v}} \leq \beta \} $;
			\item $ \wFun_{i}'(s,\vec{v}) = \wFun_i(s) + \vec{v}_{i} $;
			\item $ \wFun_{0}' = \wFun_{g} - \Vert \vec{v} \Vert $.
	\end{inparaenum}
\end{definition}

Intuitively, we are adding agent $0$ to the original game $\Game$, whose actions are $n$-dimensional vectors representing the possible rewards assigned to every other agent.
All the other components of the auxiliary game are defined accordingly.
The protocol function remains the same for every original agent, whereas the one for agent $0$ prescribes that the amount of reward distributed to the agents at each iteration does not exceeds the budget $\beta$.
The set of states is augmented to record the amount of reward received by each agent, which is then reflected in the corresponding weight function $\wFun_{i}'$.
Finally, the global weight function is updated by subtracting the amount of reward established by agent $0$ in the current iteration.

In the next two constructions, we show how to transform a $\beta$-RM for $\Game$ into a strategy for agent $0$ and viceversa.

\begin{construction}[RM to Strategy]
	\label{constr:rm-strat}
	Given a RM $ \RM = \tuple{Q_{\RM},q^0_{\RM},\delta_\RM,\tau_\RM} $ of $ \Game \implement \RM $, we define the strategy of player $ 0 $ in $ \Game' $ as $ \strElm_{\RM} = \tuple{ T_0, t_0^0, \St', \gamma_{0}, \rho_{0} } $ where
	$ T_0 = Q_{\RM}, t_0^0 = q^0_{\RM} $, and the internal transition and action functions defined as
	\begin{itemize}
		\item $ \gamma_{0}((s,\vec{v}),t) = \delta_\RM(s,t) $
		\item $ \rho_{0}((s,\vec{v}),t) = \tau_\RM(s,t) $
	\end{itemize}
	for every $ (s,\vec{v}) \in \St' $ and $ t \in T_0 $.
\end{construction}

Intuitively, the strategy $\strElm_{\RM}$ uses the same internal states of the RM $\RM$, while the transition and action functions of $\strElm_{\RM}$ are defined by modifying those of $\RM$ to match with the types required to be considered a strategy for $0$ in $\Game'$.
Such construction can be reverted by carefully modifying the types, in order to move from a strategy of agent $0$ in $\Game'$ to a RM for $\Game$, as it is shown in the following.

\begin{construction}[Strategy to RM]
	\label{constr:strat-rm}
	Given a strategy $ \strElm_{0} = \tuple{ T_0, t_0^0, \St', \gamma_{0}, \rho_{0} } $ in $ \Game' $, we define the RM for $\Game$ as $ \RM_{\strElm_{0}} = \tuple{Q_{\RM_{\strElm_{0}}},q^0_{\RM_{\strElm_{0}}},\delta_{\RM_{\strElm_{0}}},\tau_{\RM_{\strElm_{0}}}} $ where
	$ Q_{\RM_{\strElm_{0}}} = T \times \beta^{n}, q^0_{\RM_{\strElm_{0}}} = (t_0^0, \vec{0}) $, and the transition and reward functions defined as
	\begin{itemize}
		\item $ \delta_{\RM_{\strElm_{0}}}(s, (t,\vec{v})) = ( \gamma_{0}((s,\vec{v}),t), \rho_{0}((s,\vec{v}),t) ) $
		\item $ \tau_{\RM_{\strElm_{0}}}(s,(t,\vec{v})) = \rho_{0}((s,\vec{v}),t) $
	\end{itemize}
	for every $ s\in \St $ and $ (t,\vec{v}) \in Q_{\RM_{\strElm_{0}}} $.
\end{construction}

We write $ \pi_{\upharpoonright \St} $ to denote the sequence in $ \St^{\omega} $ obtained from $ \pi $ by projecting the component in $ \St $ and $ \tau(\pi) $ the sequence in $ (\SetZ^n)^{\omega} $ obtained from $ \wFun_{1}^{\RM}(\pi),\dots,\wFun_{n}^{\RM}(\pi)$.

In the following Lemma, we prove that the constructions presented above correctly translate RMs into strategies and viceversa, meaning that they make a connection between paths of $\Game \implement \RM$ and outcome of $\Game'$ when agent $0$ uses the corresponding strategy and viceversa.

\begin{lemma}
	\label{lmm:path-equivalence}
	For a given $ \Game \implement \RM $ and its associated auxiliary game $ \Game' $ the following hold: 
	
	\begin{enumerate}
		\item[(1)] for every path $ \pi \in \Paths_{\Game \implement \RM}((\instate,q^0))$, there is a path $ \pi' = (\pi_{\upharpoonright \St}, \tau(\pi)) \in \out_{\Game'}((\instate,\vec{0}), \strElm_{\RM}) $, and $ \wFun_{i}^{\RM}(\pi) = \wFun_{i}'(\pi') $ for all $ i \in \Ag $ and $ \wFun_{g}^{\RM}(\pi) = \wFun_{0}'(\pi') $;
		
		\item[(2)] for every path $ \pi' \in \out_{\Game'}((\instate,\vec{0}),\strElm_{0})$, there is a path $ \pi = (\pi_{\upharpoonright \St}', \delta_{\RM_{\strElm_{0}}}(\pi')) \in \Paths_{\Game \implement \RM_{\strElm_{0}}}((\instate,q^0)) $, and $ \wFun_{i}^{\RM_{\strElm_{0}}}(\pi) = \wFun_{i}'(\pi') $ for all $ i \in \Ag $ and $ \wFun_{g}^{\RM_{\strElm_{0}}}(\pi) = \wFun_{0}'(\pi') $.
	\end{enumerate}
	
\end{lemma}

\begin{proof}
	
	We prove the first item only, as the other has a similar proof.
	
	Observe that the path $\pi'$ is uniquely identified from $\pi$.
	Moreover, from the definitions of $\RM_{\strElm_{0}}$ and $\wFun_{0}'$, it immediately follows that $\wFun_{g}^{\RM}(\pi) = \wFun_{0}'(\pi')$.
	Therefore, we only need to prove that $\pi'$ belongs to $\out_{\Game'}((\instate,\vec{0}), \strElm_{\RM})$.
	We do it by induction on $k \in \SetN$ by showing that every prefix $\pi'_{\leq k}$ of $\pi'$ can be extended compatibly with $\strElm_{\RM}$.
	
	For the base case $k = 0$, we have that $\pi'_{\leq 0} = (s^{0}, \vec{0} )$ which is trivially extendable to any path in $\out_{\Game'}((\instate,\vec{0}), \strElm_{\RM})$.
	
	For the induction case $k > 0$, assume that $\pi'_{\leq k}$ is extendable to a path in $\out_{\Game'}((\instate,\vec{0}), \strElm_{\RM})$.
	Then, consider $\pi_{\leq k}$ and $\jact$ the joint action such that $\pi_{k + 1} = \trnFun(\pi_{k}, \jact)$, which exists since $\pi \in \Paths_{\Game \implement \RM}((\instate,q^0))$.
	Clearly, it holds that $\pi_{\geq k + 1}' = \trnFun'(\pi_{\geq k}', (\strElm_{\RM}(\pi_{\geq k}', \jact)))$, which makes $\pi_{\leq k + 1}'$ extendable compatibly with $\strElm_{\RM}$.
%
%
\end{proof}

Similarly to the correspondence between RMs and agent~$0$'s strategies, a connection between strategies of any other agent $i$ in $\Game \implement \RM$ and $\Game'$ exists.
In other words, once a RM $\RM$ and its corresponding strategy $\strElm_{\RM}$ are fixed, every strategy $\strElm_{i}$ for agent $i$ in $\Game \implement \RM$ can be translated into a strategy $\strElm_{i}'$.

\begin{construction}[$ \Game \implement \RM $ to $ \Game' $]
	\label{con:strat-rm-to-auxiliary}
	For a game $ \Game \implement \RM $ and a strategy $ \strElm_{i} = \tuple{T_i, t_i^0, \gamma_{i}, \rho_{i}} $ in it, we define a  strategy $ \hat{\strElm_{i}} = \tuple{\hat{T_i}, \hat{t}_i^0, \hat{\gamma_{i}}, \hat{\rho_{i}}} $ in the corresponding game $ \Game' $ as follows:
	\begin{itemize}
		\item $ \hat{T_i} = T \times Q^{\RM} $, and $ \hat{t}_i^0 = (t_i^0,q_0^{\RM}) $;
		\item $ \hat{\gamma_{i}}: (T \times Q^{\RM}) \times (\St \times \beta^n) \to (T \times Q^{\RM}) $ such that $ \hat{\gamma_{i}}((t,q), (s,\vec{v})) = (\gamma_{i}(t,(s,q)), \delta_\RM(s,q)) $;
		\item $ \hat{\rho_{i}}: (T \times Q^{\RM}) \times (\St \times \beta^n) \to (T \times Q^{\RM}) $ such that $ \hat{\rho_{i}}((t,q), (s,\vec{v})) = \rho_{i}(t,(s,q)) $.
	\end{itemize}

	By $\theta_{\Game \implement \RM}(\strElm_{i}) = \hat{\strElm_{i}}$ we denote the strategy for player $i$ in $\Game'$ obtained from $\strElm_{i}$ by applying the construction above.
\end{construction}

On the other hand, once a strategy $\strElm_{0}$ for agent $0$ in $\Game'$ and the corresponding RM $\RM_{\strElm_{0}}$ are fixed, the translation from strategies for agent $i$ in $\Game'$ to strategies in $\Game \implement \RM_{\strElm_{0}}$ is possible.

\begin{construction}[$ \Game' $ to $ \Game \implement \RM $]
	\label{con:strat-auxiliary-to-rm}
	For a game $ \Game' $ and a strategy $ \hat{\strElm_{i}} = \tuple{\hat{T_i}, \hat{t}_i^0, \hat{\gamma_{i}}, \hat{\rho_{i}}} $, we define a strategy $ \strElm_{i} = \tuple{T_i, t_i^0, \gamma_{i}, \rho_{i}} $, in the corresponding game $ \Game \implement \RM $ as follows:
	\begin{itemize}
		\item $ \hat{T_i} = T \times Q^{\RM} $, and $ \hat{t}_i^0 = (t_i^0,q_0^{\RM}) $;
		\item $ \hat{\gamma_{i}}: (T \times Q^{\RM}) \times (\St \times \beta^n) \to (T \times Q^{\RM}) $ such that $ \hat{\gamma_{i}}((t,q), (s,\vec{v})) = (\gamma_{i}(t,(s,q)), \delta_\RM(s,q)) $;
		\item $ \hat{\rho_{i}}: (T \times Q^{\RM}) \times (\St \times \beta^n) \to (T \times Q^{\RM}) $ such that $ \hat{\rho_{i}}((t,q), (s,\vec{v})) = \rho_{i}(t,(s,q))$.
	\end{itemize}

	By $\theta_{\Game'}(\hat{\strElm_{i}}) = \strElm_{i}$ we denote the strategy for player $i$ in $\Game \implement \RM$ obtained from $\hat{\strElm_{i}}$ by applying the construction above.
\end{construction}

The following two lemma shows that the connection among strategies in between the games also preserves the payoff of agents.

\begin{lemma}
	\label{lmm:strategy-to-auxiliary}
	For a given game $\Game$, RM $\RM$, and strategy profile $\strpElm \in \StrSet(\Game \implement \RM)$, it holds that
	
	\begin{center}
		$\pay_i^{\Game \implement \RM}(\strpElm) = \pay_i^{G'}(\strElm_{\RM}, \theta_{\Game \implement \RM}(\strpElm))$
	\end{center}
	
\end{lemma}

\begin{proof}
	Observe that the path $\pi = \pi(\strpElm, (\instate, q^0))$ belongs to the set $\Paths_{\Game \implement \RM}((\instate,q^0))$.
	Moreover, by Construction~\ref{con:strat-rm-to-auxiliary}, the path $\pi' = \pi((\strElm_{\RM}, \theta_{\Game \implement \RM}(\strpElm)), (\instate, q^0))$ is exactly the one such that $\wFun_{g}^{\RM}(\pi) = \wFun_{0}'(\pi')$ as proved in the Item 1 of Lemma~\ref{lmm:path-equivalence}.
	This straightforwardly shows that $\pay_i^{\Game \implement \RM}(\strpElm) = \pay_i^{G'}(\strElm_{\RM}, \theta_{\Game \implement \RM}(\strpElm))$.
\end{proof}

\begin{lemma}
	\label{lmm:auxiliary-to-strategy}
	For a given game $\Game$, a strategy $\strElm_{0} \in \StrSet_{0}(\Game')$, and strategy profile $\strpElm \in \StrSet^{\Game \implement \RM}$, it holds that
	
	\begin{center}
		$\pay_i^{\Game'}(\strElm_{0}, \strpElm) = \pay_i^{G \implement \RM_{\strElm_{0}}}(\theta_{\Game'}(\strpElm))$
	\end{center}
	
\end{lemma}

\begin{proof}[Sketch]
	The proof is similar to the one of Lemma~\ref{lmm:strategy-to-auxiliary}, with the use of Construction~\ref{con:strat-auxiliary-to-rm} and Item 2 of Lemma~\ref{lmm:path-equivalence}.
\end{proof}

By having the same set of payoffs, it simply follows from Lemma~\ref{lmm:path-equivalence}, Lemma~\ref{lmm:strategy-to-auxiliary}, and Lemma~\ref{lmm:auxiliary-to-strategy}, that the games $\Game \implement \RM$ and $\Game'$, where agent $0$ is bound to the use of $\strElm_{\RM}$ share the same set of Nash Equilibria.

\begin{theorem}
	\label{thm:nashinvariance}
	For a given game $\Game$ and a budget $\beta$, the two following hold:
	
	\begin{enumerate}
		\item
			For every $\beta$-RM $\RM$ and strategy profile $\strpElm$ in $\Game \implement \RM$, it holds that 
			
			\begin{center}
				$\strpElm \in \NE(\Game \implement \RM)$ iff $(\strElm_{\RM}), \hat{\strpElm}) \in \NE_{0}(\Game')$.
			\end{center}
		
		\item 
			For every strategy profile $(\strElm_{0}, \strpElm)$ in $\Game'$, it holds that
			
			\begin{center}
				$(\strElm_{0}, \strpElm) \in \NE_{0}(\Game')$ iff $\theta_{\Game \implement \RM_{\strElm_{0}}}(\strpElm) \in \NE(\Game \implement \RM_{\strElm_{0}})$
			\end{center}
				 
	\end{enumerate}	
\end{theorem}


%% file: algo.tex
\section{Solving Improvement Problems}
\label{sec:solving}

In this section, we present a technique for solving the weak and strong improvement problems. We also demonstrate how to synthesise the RM, if it exists. With the definition of the improvement problems, it makes sense to start with the problems of computing $ \worstNE $ and $ \bestNE $.
To this end, we introduce \textit{NE threshold problem}~\cite{UW11} that we will use as a subroutine in our algorithms. This problem asks whether there exists a NE in $ \Game $, such that the payoffs for the players fall between two vectors $ \vec{x} $ and $ \vec{y} $.

\begin{definition}[NE Threshold Problem]
	\label{def:ne-thresh}
	Given a game $ \Game $ and vector $ \vec{x}, \vec{y} \in (\Rat \cup \{ \pm \infty \})^n $, decide whether there is $ \strpElm \in \NE(\Game) $ with $ x_i \leq \pay_{i}(\strpElm) \leq y_i $ for every $ i \in \Ag $.
\end{definition}
When the players have pure strategies, the NE threshold problem can be solved in \np \cite{UW11}.

We begin with the following observation. For a given game $ \Game $, it holds that $ \MinW_{g}^{\Game} \leq \worstNE(\Game) \leq \MaxW_{g}^{\Game} $ and $ \MinW_{g}^{\Game} \leq \bestNE(\Game) \leq \MaxW_{g}^{\Game} $. Moreover, it also holds that for a given $ \Game' $, we have $ \MinW_{0}^{\Game'} \leq \worstNE(\Game') \leq \MaxW_{0}^{\Game'} $ and $ \MinW_{0}^{\Game'} \leq \bestNE(\Game') \leq \MaxW_{0}^{\Game'} $. As such, by using binary search and the NE threshold problem subroutine, we can compute the values of $ \worstNE $ and $ \bestNE $ for $ \Game $ and $ \Game' $.



\begin{figure}[t]
        \centering
	\RMExampletwo
	\caption{Reward machine $ \RM' $.}
	\label{fig:rm2}
\end{figure}

%
%
%
%

As we previously discussed, the optimal values of $ \worstNE(\Game) $ and $ \worstNE(\Game \implement \RM) $ may not be achievable with finite-state strategies and RMs. To see this, consider again \Cref{ex:1}. Suppose that we have a RM $ \RM' $ shown in \Cref{fig:rm2}, where $ \tau^{\RM'}_{1}(q_5,m) = 1 $ and $ \tau^{\RM'}_{1}(q,s) = 0 $ for all $ (q,s) \neq (q_5,m) $. Intuitively, player 1 is only given a reward of 1 after it finishes two cycles of deliveries. Clearly the set of NE still corresponds to the same sequence $ p(s,l,m)^{\omega} $. However, since now the designer only needs to pay 1 unit for every two cycles, we have $ \worstNE(\Game \implement \RM') = \frac{5}{6} $, which is strictly greater than $ \worstNE(\Game \implement \RM) = \frac{2}{3} $ obtained by the RM in \Cref{fig:rm}. In fact, we can increase the number of cycles needed to be done before giving 1 unit of reward by adding more states in the RM, thus obtaining strictly greater $ \worstNE $ value. Since the size of RM is not bounded, we can do this indefinitely. 
A similar argument can also be given for the optimal value of $ \worstNE(\Game) $, the complete explanation can be found in \Cref{sec:optimal}. Observe that by multiplying $ \pay_{g} $ with $ -1 $, we can also use the example above to analogously reason about $ \bestNE $.

The above arguments shows that the binary search for computing the values of $ \worstNE $ and $ \bestNE $ may not terminate. To ensure termination, we compute approximate values instead. 

\begin{definition}
	Given $ \epsilon > 0 $, an approximate value of $ \worstNE $ (resp. $ \bestNE $) is a value $ a $ such that $ a - \epsilon < o $, where $ o $ is the optimal value of $ \worstNE $ (resp. $ \bestNE $). We refer to such an approximate value as $ \epsilondash\worstNE $ (resp. $ \epsilondash\bestNE $).
\end{definition}

\begin{algorithm}[t]
	\caption{Computing $ \epsilondash\worstNE $}
	\begin{algorithmic}[1]
		\Statex \textbf{input:} $ \Game, \epsilon  $
		\State $ a_1 \leftarrow \MinW_{g}^{\Game}; a_2 \leftarrow \MaxW_{g}^{\Game}$ 
		\While{$ a_2 - a_1 \geq {\epsilon} $}
		\State $ a' \leftarrow \frac{a_1 + a_2}{2}; $
		\If{$ \exists \strpElm \in \NE(\Game), a_1 \leq \pay_{g}(\strpElm) \leq a' $}
		\State $ a_2 \leftarrow a' $
		\Else
		\State $ a_1 \leftarrow a' $
		\EndIf
		\EndWhile
		
		\State \Return $ a_2 $
	\end{algorithmic}
	\label{alg:worstNE}
\end{algorithm}

We provide \Cref{alg:worstNE} for computing $ \epsilondash\worstNE$ given $ \Game $ and $ \epsilon $ encoded in binary. 
%
The check in Line 4 corresponds to the NE threshold problem from \Cref{def:ne-thresh}. Notice that the threshold vectors $ \vec{x}, \vec{y} $ are not explicitly given, as we are not interested in these values. Thus, we fix $ x_i = \MinW_{i}^{\Game}, y_i = \MaxW_{i}^{\Game} $ for each $ i \in \Ag $, i.e., they can be of any possible values.
On the other hand, we are interested in $ \pay_{g} $, which in fact does not correspond to the payoff of any player. However, we can easily modify the underlying procedure for solving the problem in \cite{UW11} to handle this. Specifically, by \cite[Lemmas 14 and 15]{UW11}, we can specify an additional linear equation corresponding to the value of $ \pay_{g} $ being in between $ a_1 $ and $ a' $, thus yielding a procedure that is also in \np.
\Cref{alg:worstNE} can also be used to compute $ \epsilondash\worstNE(\Game') $ with the following adaptation: Line 4 is slightly modified into $ \exists \strpElm \in \NE_{0}(\Game), a_1 \leq \pay_{0}(\strpElm) \leq a' $, that is, the NE set corresponds to the $ 0 $-fixed NE. Just as with $ \pay_{g} $, $ \pay_{0} $ is not the payoff of any player in $ \Ag $. Therefore, we modify the underlying procedure for the NE threshold problem using the same approach as the above.


To compute $ \epsilondash\bestNE $, we can employ a similar technique. We make the following modification to \Cref{alg:worstNE}: in each iteration, instead of checking the left-half part, we check the right-half part (i.e., instead of minimising, we are \textit{maximising}). This is done in Lines 4-8 of the algorithm by checking whether $ \exists \strpElm \in \NE(\Game), a' \leq \pay_{g}(\strpElm) \leq a_2 $. If the check returns true, we set $ a_1 \leftarrow a' $, otherwise $ a_2 \leftarrow a' $. Again, as with $ \worstNE $, we slightly modify Line 4 in order to compute $ \bestNE(\Game') $.

\begin{theorem}
	\label{thm:compute-best-worst}
	Given a game $ \Game $ (resp. $ \Game' $) and $ \epsilon > 0 $, the problems of computing $ \epsilondash\bestNE(\Game) $ and $\epsilondash\worstNE(\Game) $ (resp. $ \epsilondash\bestNE(\Game') $ and $\epsilondash\worstNE(\Game') $) are $ \FP^{\np} $-complete.
\end{theorem}

\begin{proof}
	The upper bounds follows from \Cref{alg:worstNE}. The while loop runs in polynomial number of steps (i.e., logarithmic in $ |\Game| \cdot \frac{1}{\epsilon} $), and in each step calls a \np oracle. Observe that $ \epsilon $ can be arbitrarily small (i.e., arbitrary precision). For the lower bound we reduce from \textsc{TSP Cost} which is $ \FP^{\np} $-hard~\cite{1994-papadimitriou}. Given a \textsc{TSP Cost} instance $ (G,c), G = (V, E) $ is a graph, $ c : E \to \SetZ $ is a cost function, we construct a game $ \Game $ such that the $ \epsilondash\worstNE(\Game) $ corresponds to the value of optimum tour\footnote{For auxiliary game $ \Game' $, we can easily adapt the reduction by substituting $ \wFun_{g} $ with $ \wFun_{0} $.}. Let $ \Game $ be such a game where
	\begin{itemize}
		\item $ \Ag = V $,
		\item $ \St = \{ (e,v) : e \in E \wedge v = \trg(e) \} \cup \{ (\star,\sink) \} $,
		\item $ s^0 $ can be chosen arbitrarily from $ \St \setminus \{ (\star,\sink) \} $,
		\item for each state $ (e,v) \in \St $ and each player $ i \in \Ag $
		\begin{itemize}
			\item $ \protocol_i((e,v)) = \{ \out(v) \} \cup \{\star\} $ if $ i = v $
			\item $ \protocol_i((e,v)) = \{\circ,\star\} $, otherwise;
		\end{itemize}
		\item for each state $ (e,v) \in \St $ and action profile $ \AcProf $
		\begin{itemize}
			\item $ \trnFun((e,v),\AcProf) = (\act_{v},\trg(\act_{v})) $ if $ v \neq \sink $ and $ \forall i \in \Ag, \act_{i} \neq \star $;
			\item $ \trnFun((e,v),\AcProf) = (\star, \sink) $, otherwise;
		\end{itemize}
		\item for each state $ (e,v) \in \St $ and player $ i \in \Ag $
		\begin{itemize}
			\item $ \wFun_{i}((e,v)) = |V| $, if $ v = i $ and $ v \neq \sink $,
			\item $ \wFun_{i}((e,v)) = 0 $, if $ v \neq i $ and $ v \neq \sink $,
			\item $ \wFun_{i}((e,v)) = 1 $, if $ v = \sink $;
		\end{itemize}
		\item  for each state $ (e,v) \in \St $
		\begin{itemize}
			\item $ \wFun_{g}((e,v)) = \max \{ c(e'): e' \in E \} \cdot |V| $, if $ v = \sink $
			\item $ \wFun_{g}((e,v)) = c(e) \cdot |V| $, otherwise;
		\end{itemize}		
	\end{itemize}
	where $ \circ, \star, \sink $ are fresh symbols. We also set $ \epsilon = 1 $. The construction is complete and polynomial to the size of $ (G, c) $.
	
	We argue that $ \lfloor \epsilondash\worstNE(\Game) \rfloor $ is exactly the value of optimal valid tour. First, observe that for any $ \strpElm \in \NE(\Game) $, it holds that either (1) $ \pi(\strpElm) $ visits every $ v \in V $ (i.e., visits every city), thus a valid tour, or (2) $ \pi(\strpElm) $ enters $ (\star,\sink) $ and stays there forever. Case (1) holds because if $ \pi(\strpElm) $ does not visit $ v \in V $, then $ \pay_{v}(\strpElm) = 0 $ thus player $ v $ will deviate to $ (\star,\sink) $ and obtain better payoff. In fact, $ \strpElm $ visits each city exactly once, because otherwise, there is a player who gets payoff strictly less than 1, and deviates to $ (\star,\sink) $. Case (2) is trivially true; however, assuming that the costs are not uniform (otherwise \textsc{TSP Cost} becomes trivial), it cannot be a solution to $ \epsilondash\worstNE $. Let $ o $ be the optimal tour cost, and suppose for a contradiction that $ \lfloor \epsilondash\worstNE(\Game) \rfloor < o $. Let $ \hat{\strpElm} $ be a corresponding strategy profile. By the construction of $ \Game $, this means that $ \hat{\strpElm} $ does not visit some cities or visits some cities more than once. However, by (1) above, $ \hat{\strpElm} $ cannot be in $ \NE(\Game) $---a contradiction. We can argue in a similar manner for $ \lfloor \epsilondash\worstNE(\Game) \rfloor > o $; it is not possible because either the corresponding strategy does not form a valid tour (and by (1) above, it is not a NE), or it is not the optimal solution to $ \epsilondash\worstNE $; again a contradiction. Finally, since $ \epsilondash\worstNE $ approaches $ \worstNE $ from the right, we have $ \lfloor \epsilondash\worstNE(\Game) \rfloor = o$.
	
	For $ \bestNE $, we can use the same construction but with the following modification to $ \wFun_{g} $:
	\begin{itemize}
		\item $ \wFun_{g}((e,v)) = -(\max \{ c(e'): e' \in E \} \cdot |V|) $, if $ v = \sink $
		\item $ \wFun_{g}((e,v)) = -(c(e) \cdot |V|) $, otherwise;
	\end{itemize}
	and use similar argument as the above.
\end{proof}

\paragraph*{Approximate improvement problems}

We define the approximate improvement problems as follows.

\begin{definition}[$ \epsilon $-improvement problem]
	Given a game $ \Game $, a budget $ \beta $, a threshold $ \Delta $, and $ \epsilon $. The $ \Gamma $ $ \epsilon $-improvement problem, with $ \Gamma \in \{\text{strong}, \text{weak}\} $, decides whether there exists a $ \beta $-RM $ \RM $ such that:
	\[ \epsilondash\GammaNE(\Game \implement \RM) - \epsilondash\GammaNE(\Game) > \Delta. \]
\end{definition}

Having the procedures for computing $ \epsilondash\worstNE $ and $ \epsilondash\bestNE $ for both $ \Game $ and $ \Game' $, we can then directly solve the $ \epsilon $-improvement problem with the following procedure.

\begin{enumerate}
	\item Build the auxiliary game $ \Game' $;
	\item Compute $ \epsilondash\GammaNE(\Game) $ and $ \epsilondash\GammaNE(\Game') $;
	\item If  $ \epsilondash\GammaNE(\Game') -  \epsilondash\GammaNE(\Game) > \Delta $, then return ``yes''; otherwise return ``no''.
\end{enumerate}

\begin{theorem}
	Strong and weak $ \epsilon $-improvement problems are $ \DeltaPTwo $.
\end{theorem}

\begin{proof}
	The upper bounds follow from the procedure described above. Steps 1 and 3 can be done in polynomial time, Step 2 only needs two calls to an $ \FP^{\np} $ oracle. Thus we have a decision procedure that runs in $ \ptime^{\np} = \DeltaPTwo $.
\end{proof}

\begin{theorem}
	Strong and weak $ \epsilon $-improvement problems are $ \np $-hard and $ \conp $-hard, respectively.
\end{theorem}

\begin{proof}
	To show that strong $ \epsilon $-improvement problem is $ \np $-hard,
	we reduce from \textsc{Hamiltonian Path} problem: given a directed graph $ G = (V,E) $, is there a path that visits each vertex exactly once; this problem is $ \np $-hard~\cite{1994-papadimitriou}. We build a game $ G $ and fix $ \beta, \Delta$  and $ \epsilon $ such that the strong $ \epsilon $-improvement problem returns yes if and only if \textsc{Hamiltonian Path} returns yes. Given a \textsc{Hamiltonian Path} instance $ G = (V, E) $, we construct a game $ \Game $ as follows.
	\begin{itemize}
		\item $ \Ag = V \cup \{n+1,n+2\} $, where $ V = \{1,...,n\} $,
		\item $ \St = \{ (e,v) : e \in E \wedge v = \trg(e) \} \cup \{ (\star,\sink), (\star, \blacksquare), (\star, \triangle) \} $,
		\item $ \instate $ can be chosen arbitrarily from $ \St \setminus \{ (\star,\sink), (\star, \blacksquare), (\star, \triangle) \} $,
		\item for each state $ (e,v) \in \St $ and each player $ i \in \Ag $
		\begin{itemize}
			\item $ \protocol_i((e,v)) = \{ \out(v) \} \cup \{\star\} $ if $ i = v $
			\item $ \protocol_i((e,v)) = \{\circ,\star\} $, otherwise;
		\end{itemize}
		\item for each state $ (e,v) \in \St $ and action profile $ \AcProf $
		\begin{itemize}
			\item $ \trnFun((e,v),\AcProf) = (\act_{v},\trg(\act_{v})) $ if $ v \neq \sink $ and $ \forall i \in V, \act_{i} \neq \star $;
			\item $ \trnFun((e,v),\AcProf) = (\star, \sink) $, if $ v \neq \sink $ and $ \exists i \in V, \act_{i} = \star $;
			\item $ \trnFun((e,v),\AcProf) = (\star, \blacksquare) $, if $ v = \sink $ and $ \act_{n+1} = \act_{n+2} $;
			\item $ \trnFun((e,v),\AcProf) = (\star, \triangle) $, if $ v = \sink $ and $ \act_{n+1} \neq \act_{n+2} $;
			\item $ \trnFun((e,v),\AcProf) = (\star, v) $, if $ v \in \{ \blacksquare, \triangle \} $;
		\end{itemize}
		\item for each state $ (e,v) \in \St $ and player $ i \in \{1,...,n\} $
		\begin{itemize}
			\item $ \wFun_{i}((e,v)) = |V| $, if $ v = i $ and $ v \not\in \{\sink,\blacksquare,\triangle\} $,
			\item $ \wFun_{i}((e,v)) = 0 $, if $ v \neq i $ and $ v \not\in \{\sink,\blacksquare,\triangle\} $,
			\item $ \wFun_{i}((e,v)) = 1 $, if $ v \in \{\sink,\blacksquare,\triangle\} $;
		\end{itemize}
	\item for each state $ (e,v) \in \St $ and player $ i \in \{n+1,n+2\} $
	\begin{itemize}
		\item $ \wFun_{i}((e,v)) = 0 $;
	\end{itemize}
		\item  for each state $ (e,v) \in \St $
		\begin{itemize}
			\item $ \wFun_{g}((e,v)) = 0 $, if $ v \in \{\sink,\blacksquare,\triangle\} $
			\item $ \wFun_{g}((e,v)) = |V| $, otherwise;
		\end{itemize}		
	\end{itemize}
where $ \circ, \star, \sink, \blacksquare,\triangle $ are fresh symbols. We also set $ \beta = 1, \epsilon = 1, \Delta = \frac{1}{2} $. The construction is complete and polynomial to the size of $ G $.

Observe that $ \worstNE(\Game) = 0 $, where the play goes to either $ (\ast, \blacksquare) $ or $ (\ast, \triangle) $ and stays there forever. However, with $ \beta = 1 $, the designer can pay player $ n+1 $ (resp. player $ n+2 $) with a payment of $ 1 $ when the play reaches $ (\ast,\blacksquare) $ (resp. $ (\ast, \triangle) $). Essentially, forcing $ n+1 $ and $ n+2 $ to play a matching pennies game, a game with no Nash equilibrium. Thus, the play that goes to either $ (\ast, \blacksquare) $ or $ (\ast, \triangle) $ no longer part of $ \NE(\Game) $. Now, consider a run that visits each $ v \in V $ exactly once, this is a Nash equilibrium. The reasoning is the same as the one provided in the proof of \Cref{thm:compute-best-worst}. And by construction, such a run can only be possible if and only if there is a Hamiltonian path in the corresponding graph $ G $. Let $ \pi $ be such a run, now we have $ \pay_{g}(\pi) = 1 $ and thus, $ \epsilondash\worstNE(\Game \implement \RM) - \epsilondash\worstNE(\Game) = 1 - 0 > \frac{1}{2} $.

The proof for weak $ \epsilon $-improvement problem is similar: through a reduction from the complement of \textsc{Hamiltonian Path}. The proof is included in \Cref{appendix:hamilton}.
\end{proof}

\paragraph*{Synthesis of Reward Machines}
Given a game $ \Game $, a budget $ \beta $, a threshold $ \Delta $, and $ \epsilon $, if the strong (resp. weak) $ \epsilon $-improvement problem returns a positive answer, then we can synthesise the corresponding RM $ \RM $ as follows. From the auxiliary game $ \Game' $, find a strategy profile $ (\strElm_{0},\strpElm) \in \NE_{0}(\Game') $ such that $ \pay_{0}^{\Game'}(\strElm_{0},\strpElm) = \worstNE(\Game') $ (resp. $ \pay_{0}^{\Game'}(\strElm_{0},\strpElm) = \worstNE(\Game') $). Using \Cref{constr:strat-rm}, we obtain the RM $ \RM_{\strElm_{0}} $ from $ \strElm_{0} $, which corresponds to the required RM.

%% file: conclusion.tex
\section{Conclusion}

	In this paper, we examined games where each agent had a weight function over states, with their utility determined by the mean-payoff aggregation.
	A global weight function was used to gauge designer satisfaction, also measured through mean-payoff value. We utilised reward machines to enhance designer satisfaction, reconfiguring weights after each iteration to reshape the equilibrium set.
	Our aim was to boost the global payoff over equilibria by at least a given value $\Delta$, achieved by strategically synthesising a suitable reward machine.

	Among the other results, we first demonstrated that reward machines are strictly more effective than subsidy schemes.
	However, we also found that in some cases, although no reward machine could improve the global payoff by the required value $\Delta$, an $\epsilon$-approximation could be found.
	Thus, we introduced and addressed the $\epsilon$-improvement problem as a more general approach to equilibrium design.

	Since multiple equilibria are possible in these games, we analysed the synthesis problem from both optimistic and pessimistic perspectives.
	We aimed to enhance the global mean-payoff over the best and worst possible Nash Equilibria, considering scenarios where agents select the most or least convenient equilibrium from the designer's viewpoint, respectively.
	We also provided complexity classifications for these problems, demonstrating that each could be solved in $\DeltaPTwo$ and were at least $\np$-hard and $ \conp $-hard.

	\paragraph{Future work}
	Several directions are possible from this.
	First, extensions of designers and agents' objectives should be considered.
	For example, in~\cite{GNPW19,GMPRW17} the agents' goals are represented as a combination of \LTL and mean-payoff objectives, arranged in a lexicographic fashion.
	Also multi-valued logic such as LTL[F] are considered for rational verification~\cite{AKP18}.
	It would be interesting to find out how to employ reward machines to boost the satisfaction value for this case.
	Last but not least, an excursion into normative systems should be considered.
	Although dynamic norms as defined in~\cite{huang16} are of the same type of reward machines~\cite{icarte2022reward}, their implementation to games provide very different effects.
	On the one hand, norms disable agents' actions.
	On the other hand, reward machines do not strictly forbid agents to execute their actions in the game, but rather reward-incentivise those that are more convenient from the global standpoint.
	It would be interesting to combine the two approaches, finding the right balance between obligation and recommendation modalities.

%% file: supplement.tex
\section{Appendix}
\label{sec:appendix}

\subsection{On Exact Optimal $ \worstNE(\Game) $}
\label{sec:optimal}
Exact optimal value of $ \worstNE(\Game) $ may not be achievable with finite-memory strategies. To see this, consider the game arena below.

\begin{center}
	\scalebox{0.7}{
	\begin{tikzpicture}[state/.style={circle, draw, minimum size=1cm}, node distance=1.5cm]
		\node[] (axis) {};
		\node[state, label=below:{$ (0,0) $}] (s^0) [above = of axis] { $t$};
		\node[] (start) [above =0.5cm of s^0] {};
		
		\node[state, label=below:{$ (0,1) $}] (s^2) [right = of axis] { $r$};
		\node[state, label=below:{$ (1,0) $}] (s^1) [left = of axis] { $l$};
		\node[state, label=below:{$ (0,0) $}] (s^3) [below = of axis] {\small $b$};
		
		\draw [-{Latex[width=2mm]}]
		(start) edge node{} (s^0)
		(s^0) edge[] node[left, yshift=5pt]{\small $(*,*)$} (s^1)
		(s^1) edge[] node[left, yshift=-5pt]{\small $(\ast,R)$} (s^3)
		
		(s^3) edge[] node[right, yshift=-5pt]{\small $(*,*)$} (s^2)
		(s^2) edge[] node[right, yshift=5pt]{\small $(L,\ast)$} (s^0)
		%
		(s^1) edge[loop, out=200, in=160, distance=1cm] node[left]{\small $(*,L)$} (s^1)
		(s^2) edge[loop, out=20, in=340, distance=1cm] node[right]{\small $(R,*)$} (s^2)
		;
	\end{tikzpicture}
}
\end{center}

	The set of players is $ \Ag = \{1,2\} $ and the initial state is $ t $. Let $ \wFun_{g}(s) = -\wFun_{1}(s) $. Consider the run $ \pi = (t,l,b,r)^{\omega} $; this is a NE with $ \pay_{g}(\pi) = -\frac{1}{4} $. However, increasing the number of loops in $ l $ and $ r $ also decreases $ \pay_{g} $. For instance $ \pi' = (t,l,l,b,r,r)^{\omega} $ is also an NE. To see this consider a ``threat'' strategy by player 1 as follows: if player 2 does not loop at least one time in $ l $, then never loop in $ r $ forever more. We can reason similarly for player 2 to player 1. Therefore, it is a NE.
	However, we have $ \pay_{g}(\pi') = - \frac{1}{3} $, thus $ \pay_{g}(\pi) > \pay_{g}(\pi') $. We can continue increasing the number of loops, and decreasing the value of $ \pay_{g} $.

\subsection{Proof for weak $ \epsilon $-improvement problem $ \conp $-hardness}
\label{appendix:hamilton}

We reduce from the complement of \textsc{Hamiltonian Path} problem. We build a game $ G $ and fix $ \beta, \Delta$  and $ \epsilon $ such that the weak $ \epsilon $-improvement problem returns yes if and only if \textsc{Hamiltonian Path} returns no. Given a \textsc{Hamiltonian Path} instance $ G = (V, E) $, we construct a game $ \Game $ as follows.
\begin{itemize}
	\item $ \Ag = V \cup \{n+1,n+2\} $, where $ V = \{1,...,n\} $,
	\item $ \St = \{ (e,v) : e \in E \wedge v = \trg(e) \} \cup \{ (\star,\sink), (\star, \blacksquare), (\star, \triangle) \} $,
	\item $ \instate $ can be chosen arbitrarily from $ \St \setminus \{ (\star,\sink), (\star, \blacksquare), (\star, \triangle) \} $,
	\item for each state $ (e,v) \in \St $ and each player $ i \in \Ag $
	\begin{itemize}
		\item $ \protocol_i((e,v)) = \{ \out(v) \} \cup \{\star\} $ if $ i = v $
		\item $ \protocol_i((e,v)) = \{\circ,\star\} $, otherwise;
	\end{itemize}
	\item for each state $ (e,v) \in \St $ and action profile $ \AcProf $
	\begin{itemize}
		\item $ \trnFun((e,v),\AcProf) = (\act_{v},\trg(\act_{v})) $ if $ v \neq \sink $ and $ \forall i \in V, \act_{i} \neq \star $;
		\item $ \trnFun((e,v),\AcProf) = (\star, \sink) $, if $ v \neq \sink $ and $ \exists i \in V, \act_{i} = \star $;
		\item $ \trnFun((e,v),\AcProf) = (\star, \blacksquare) $, if $ v = \sink $ and $ \act_{n+1} = \act_{n+2} $;
		\item $ \trnFun((e,v),\AcProf) = (\star, \triangle) $, if $ v = \sink $ and $ \act_{n+1} \neq \act_{n+2} $;
		\item $ \trnFun((e,v),\AcProf) = (\star, v) $, if $ v \in \{ \blacksquare, \triangle \} $;
	\end{itemize}
	\item for each state $ (e,v) \in \St $ and player $ i \in \{1,...,n\} $
	\begin{itemize}
		\item $ \wFun_{i}((e,v)) = |V| $, if $ v = i $ and $ v \not\in \{\sink,\blacksquare,\triangle\} $,
		\item $ \wFun_{i}((e,v)) = 0 $, if $ v \neq i $ and $ v \not\in \{\sink,\blacksquare,\triangle\} $,
		\item $ \wFun_{i}((e,v)) = 1 $, if $ v \in \{\sink,\blacksquare,\triangle\} $;
	\end{itemize}
	\item for each state $ (e,v) \in \St $
	\begin{itemize}
		\item $ \wFun_{n+1}((e,v)) = 1 $ if $ v = \blacksquare $;
		\item $ \wFun_{n+1}((e,v)) = 0 $ otherwise;
	\end{itemize}
	\item for each state $ (e,v) \in \St $
	\begin{itemize}
		\item $ \wFun_{n+2}((e,v)) = 1 $ if $ v = \triangle $;
		\item $ \wFun_{n+2}((e,v)) = 0 $ otherwise;
	\end{itemize}
	\item  for each state $ (e,v) \in \St $
	\begin{itemize}
		\item $ \wFun_{g}((e,v)) = 0 $, if $ v \in \{\sink,\blacksquare\} $
		\item $ \wFun_{g}((e,v)) = 2 $, if $ v = \triangle $
		\item $ \wFun_{g}((e,v)) = |V| $, otherwise;
	\end{itemize}		
\end{itemize}
where $ \circ, \star, \sink, \blacksquare,\triangle $ are fresh symbols. We also set $ \beta = 1, \epsilon = 1, \Delta = \frac{1}{2} $. The construction is complete and polynomial to the size of $ G $.

Observe that when the play goes to $ (\ast,\sink) $, players $ n+1 $ and $ n+2 $ is forced to play matching pennies game. Thus, any run that ends up in $ (\ast,\sink) $ cannot be in $ \NE(\Game) $. Now consider a RM $ \RM $ where player $ n+1 $ is paid a payoff of $ 1 $ when the play reaches $ (\ast,\triangle) $. Let $ \pi $ be a run that ends up in $ (\ast,\triangle) $. We have that $ \pi \in \NE(\Game \implement \RM) $ since $ n+1, n+2 $ do not play matching pennies game any more, and thus $ \pay_{g}(\Game \implement \RM) = 1$. Notice that $ \pay_{g}(\Game) = 0 $ if and only if $ \NE(\Game) = \varnothing $, which can only happen if there is no Hamiltonian path in $ G $; otherwise, the players $ \{1,...,n\} $ can play a run $ \pi' $ that visits each vertex $ v \in V $ exactly once such that for each player $ i \in \Ag $, $ \pay_{i}(\Game) = 1 $, thus corresponding to a Nash equilibrium (see proof of \Cref{thm:compute-best-worst}). Consequently, $ \epsilondash\bestNE(\Game \implement \RM) - \epsilondash\bestNE(\Game) > \frac{1}{2}$ if and only if there is no Hamiltonian path in $ G $.